\documentclass[12pt]{article}
\usepackage{cite}
\usepackage{amsmath,amssymb,amsfonts}
\usepackage{algorithmic}
\usepackage{graphicx}
\usepackage{algorithm,algorithmic}
\usepackage{hyperref}

\newtheorem{definition}{Definition}
\newtheorem{proposition}{Proposition}
\newtheorem{remark}{Remark}
\newtheorem{theorem}{Theorem}
\newtheorem{proof}{Proof}

\begin{document}

\title{Analysis and Control of Perturbed Density Systems}

\author{\textbf{Igor B. Furtat}
\\
\\
Institute for Problems in Mechanical Engineering
\\
of the Russian Academy of Sciences (IPME RAS), 
\\
199178, Saint-Petersburg, Boljshoy prospekt V.O., 61, Russia 
\\
e-mail: cainenash@mail.ru}

\maketitle

\begin{abstract}
The paper investigates dynamical systems for which the derivative of some positive-definite function along the solutions of this system depends on so-called density function. 
In turn, such dynamical systems are called density systems. 
The density function sets the density of the space, where the system is evolved, and affects the behaviour of the original system. 
For example, this function can define (un)stable regions and forbidden regions where there are no system solutions.
The density function can be used in the design of new adaptive control laws with the formulation of appropriate new control goals, e.g., stabilization in given bounded or semi-bounded sets. 
To design a novel adaptive control law that ensures the system outputs in a given set, systems with known and unknown parameters under disturbances are considered. 
All theoretical results and conclusions are illustrated by numerical simulations.
\end{abstract}

Keywords: dynamical system, positive-definite function, stability, instability, control, adaptive control.

\section{Introduction}

The paper describes dynamical systems whose behaviour depends on a so-called density function. 
Such dynamical systems are called density systems.
The density function determines the density of the space, where the system is evolved, affects the behaviour of phase trajectories and transients. 
Definitions and properties of density function and density systems are given in the main part of the paper.

A particular class of the density systems is first considered in \cite{Krasnoselski63}, where to study the stability of $\dot{x}=f(x)$ by using Lyapunov function a new system 
$\dot{x}=\rho(x)f(x)$ is considered with the auxiliary function $\rho(x)>0$. 
This idea is then used in \cite{Zhukov79,Zhukov99} to study stability and instability using the properties of divergence and flow of the phase vector. 
The stability study problem is considered in \cite{Zhukov99} for second order systems. 
The convergence of almost all trajectories of arbitrary order systems is proposed in \cite{Rantzer01}. 
Differently from \cite{Krasnoselski63,Zhukov79,Zhukov99}, the function $\rho(x)>0$ is called the density function in \cite{Rantzer01}.
The stability and instability of arbitrary order systems is proposed in \cite{Furtat20a,Furtat20b,Furtat21,Furtat22}. 
It is shown in \cite{Furtat20a,Furtat20b,Furtat21,Furtat22} the relationship between the obtained stability criteria, depending on the density function, and the continuity equation in electromagnetism \cite{Griffiths17}, fluid dynamics \cite{Pedlosky79,Potter08}, energy and heat \cite{Arnold14}, and quantum mechanics \cite{McMahon13}. 
New methods of the control law design are also proposed in \cite{Rantzer01,Furtat20a,Furtat20b,Furtat21,Furtat22}, using the density function and some properties of the divergence of the phase vector.

Another special class of density systems is considered in control law design by using the second Lyapunov method, where different functions are used in order to keep system outputs in given sets. 
These functions are regarded in our paper as density functions.
The papers \cite{Liberzon13,Berger18,Hu22} proposes the control laws ensuring the outputs in given funnels. 
In \cite{Bechlioulis14,Bikas22} the prescribed performance control laws guarantee the finding of transients in a set converging to a neighbourhood of zero. 
The papers \cite{Orlov22a,Orlov22b,Su23} consider the stabilization of trajectories in a prescribed-time using an infinitely increasing coefficient in the control law.
The approaches \cite{Tee09,Tian22} are based on the use of barrier or penalty functions to keep the system outputs within given constraints.
The control laws \cite{Furtat21b,Furtat21c} allow to guarantee the location of the outputs in the set, which can be asymmetric with respect to the equilibrium position and does not converge to a constant. 

In the above mentioned results, the density functions are positive, continuously differentiable, and multiplied by the entire right hand side of the system under consideration. 
In \cite{Furtat23,Furtat23ait,Furtat24cdc}, the systems for which the derivative of a positive definite function along the solutions of these systems is multiplied by the density function, but not the original systems. 
Thus, the density function can be represented explicitly or implicitly on the right hand side of the system equation, which extends the class of investigated systems. 
Also, sign-indefinite and nonsmoth density functions can be considered.
Using these functions, the density of the space is determined in sense of the selection of (un)stable regions, forbidden regions (where the system has no solutions), and the regions with different values of the density that affects on the behaviour of the studied systems.

However, in \cite{Furtat23,Furtat23ait}, systems are considered without external disturbances. 
The presence of bounded disturbances can make the closed-loop system unstable in \cite{Furtat23,Furtat23ait}. 
In \cite{Furtat24cdc}, the preliminary results are given under bounded disturbances. 
The present paper expands the results \cite{Furtat24cdc} to systems with disturbances and with the presentation of new comparative results. 
Also, a new adaptive control law is obtained on the basis of the density function for systems with disturbances.

The contribution of this paper is as follows:
\begin{enumerate}
\item[(i)] differently from \cite{Krasnoselski63,Zhukov79,Zhukov99,Rantzer01,Furtat20a,Furtat20b,Furtat21,Furtat22}, we consider the systems where the density function is not necessarily represented in the entire right hand side of the system equation;
\item[(ii)] in contrast to \cite{Liberzon13,Berger18,Hu22,Orlov22a,Orlov22b,Su23,Furtat21b,Furtat21c}, the density function can be represented implicitly on the right hand side of the system equation;
\item[(iii)] differently from \cite{Liberzon13,Berger18,Hu22,Tee09,Tian22,Furtat21b,Furtat21c}, the density function can guarantee solutions in an unbounded set with forbidden domains, and the boundaries of these sets can be piecewise continuous in time;
\item[(iv)] unlike \cite{Furtat23,Furtat23ait} a system analysis and control law design under disturbances (perturbations) are considered;
\item[(v)] compared to \cite{Furtat24cdc} a new adaptive control law is obtained on the basis of the density function for systems with disturbances. Differently from \cite{Furtat23}, the improved adaptive control law allows to be stabilize systems not necessarily in the vicinity of zero. 
\end{enumerate}

The paper is organised as follows.
Section \ref{Sec2} describes definitions of the density function and density systems with bounded disturbances.
Some properties of these systems are proved. 
Section \ref{Sec3} is dedicated to the use of the density function to design the adaptive control law for systems with unknown parameters and disturbances. 
Section \ref{Sec6} summarises the main conclusions. 
All sections include simulations to illustrate the theoretical results.

The following \textit{notations} are used in the paper:
$\mathbb R^{n}$ is an $n$-dimensional Euclidean space with norm $|\cdot|$;
$\mathbb R_{+}$ ($\mathbb R_{-}$) is a set of positive (negative) real numbers.

\section{Density systems}
\label{Sec2}

Consider a dynamical system of the form
\begin{equation}
\label{eq1}
\begin{array}{l}
\dot{x}=f(x,t),
\end{array}
\end{equation}
where $t \geq 0$, $x \in D \subset \mathbb R^n$ is the state,
the function $f: D \times [0, +\infty) \to \mathbb R^{n}$ is continuous in $t$ and piecewise continuous in $x$ on $D \times [0, +\infty)$. 
The solutions of \eqref{eq1} and other systems with discontinuous right hand side are understood in the sense of Filippov \cite{Filippov88}.


In \cite{Krasnoselski63,Zhukov79,Zhukov99,Rantzer01,Furtat20a,Furtat20b,Furtat21,Furtat22}, the autonomous system $\dot{x}=\rho(x)f(x)$ with the density function $\rho(x)>0$ is considered for the stability analysis of the initial system $\dot{x}=f(x)$, where $\rho(x)>0$ and $f(x)$ are continuously differentiable. 
It is shown in \cite{Furtat20a,Furtat20b,Furtat21,Furtat22}, that $\rho(x)$ determines the phase space density and it affects the phase flow velocity.
If $\rho(x)>0$, then the presence of the density function does not qualitatively affect the equilibrium positions and their types, but quantitatively affects the phase portrait.
Similar conclusions are drawn in \cite{Furtat20b} for nonautonomous systems $\dot{x}=f(x,t)$ with density function $\rho(x,t)>0$ for all $x$ and $t$. 

Unlike  \cite{Krasnoselski63,Zhukov79,Zhukov99,Rantzer01,Furtat20a,Furtat20b,Furtat21,Furtat22} in the present section we consider sign-indefinite and non-differentiable density functions. 
Such density functions allow one to set the properties of the space in order to change the behaviour of the original system.
Using the density function in the next section, a novel control laws are proposed, where by changing the density of the space one can achieve the corresponding control goals for the investigated systems.

\begin{remark}
Note that positive density is a usual in various fields of science.
However, the negative density is also used in some branches of science, see for example \cite{Bondi57,Norbert15,Farnes18,Edward20} in cosmology (white holes and wormholes), in general relativity, in quantum mechanics, in Schroedinger equation, in the theory of vibrations and metamaterials, in the experiment to create negative effective mass by reducing the temperature of rubidium atoms to near absolute zero, etc. 
In control law design, a density function with variable sign can be virtually implemented in a controller.
\end{remark}

Introduce definitions of the density functions and density systems considered in the present paper.


\begin{definition}
\label{def1}
The system \eqref{eq1} is called density with the density function $\rho(x,t): D \times [0, +\infty) \to \mathbb R$,
if there exists a continuously differentiable function $V(x,t): D \times [0, +\infty) \to \mathbb R_{+}$ and at least one of the following conditions holds
\begin{itemize}
\item[(a)] $\dot{V} \leq \rho(x,t) W_1(x) \leq 0$ for any $t \geq 0$ and $x \in D_{S} \subset D$,
\item[(b)] $\dot{V} \geq \rho(x,t) W_2(x) \geq 0$ for any $t \geq 0$ and $x \in D_{U} \subset D$.
\end{itemize}
Hereinafter $\rho(x,t)$ is continuous in $t$ and piecewise continuous in $x$ on $D \times [0, +\infty)$, $W_1(x)$ and $W_2(x)$ are non-zero continuous functions in $D_{S}$ and $D_{U}$ respectively.
\end{definition}


\begin{definition}
\label{def1_0}
If the functions $W_1(x)$ and $W_2(x)$ in Definition \ref{def1} are continuous 
in $D_S$ and $D_U$ respectively, then the system \eqref{eq1} is said to be weak density.
\end{definition}


\begin{definition}
\label{def3}
If condition (b) in Definition \ref{def1} satisfies $\dot{V} \leq \rho(x,t) W_1(x) < 0$ 
in $D_S \times [0, +\infty)$ or $\dot{V} \geq \rho(x,t ) W_2(x)>0$ in $D_U \times [0, +\infty)$,
then the system \eqref{eq1} is called strictly density.
\end{definition}


\begin{definition}
\label{def2}
If $\dot{V} \leq \rho(x,t) W_1(x) \leq 0$ in $D_S \times [0, +\infty)$,
then the density function $\rho(x,t)$ and the domain $D_S$ are said to be stable.
If $\dot{V} \geq \rho(x,t) W_2(x) > 0$ in $D_U \times [0, +\infty)$,
then the density function $\rho(x,t)$ and the domain $D_U$ are said to be unstable.
\end{definition}


Let us start with a simple example to illustrate the proposed definitions.


\textit{Example 1.}
Consider the system
\begin{equation}
\label{eq_ex1_00}
\begin{array}{l} 
\dot{x}(t)=-x(t)+d(t),
\end{array}
\end{equation} 
where $x \in \mathbb R$ and $d(t)$ is an external bounded disturbance with $\bar{d}=\sup\{|d(t)|\}$. 
The system \eqref{eq_ex1_00} has equilibrium point $x=0$ under $d=0$. 
To analyse the stability of \eqref{eq_ex1_00} consider Lyapunov function
\begin{equation}
\label{eq_ex1_00_Lyap}
\begin{array}{l} 
V=0.5x^2.
\end{array}
\end{equation} 
Differentiating \eqref{eq_ex1_00_Lyap} along with \eqref{eq_ex1_00}, one gets 
$\dot{V} \leq -\mu x^2-(1-\mu)x^2+\bar{d}|x| \leq -\mu x^2 <0$ in the domain $D_S=\left\{x \in \mathbb R: |x| > \frac{\bar{d}}{1-\mu}\right\}$, $0<\mu<1$. 
Thus, the system \eqref{eq_ex1_00} is stable in $D_S$.

Together with the system \eqref{eq_ex1_00}, consider a new one in the form
\begin{equation}
\label{eq_ex1_0}
\begin{array}{l}
\dot{x}(t)=-\rho(x,t) x(t) + d(t),
\end{array}
\end{equation}
where $\rho(x,t): \mathbb R \times [0, +\infty) \to \mathbb R$ is a continuous function in $t$ and piecewise continuous in $x$. 
Obviously, the behaviour of the system \eqref{eq_ex1_0} depends on the properties of $\rho(x,t)$.
With the function $\rho(x,t)$ it is possible to set some properties and restrictions in the space $(x,t)$.
Thus, it is possible to influence the quality of the transients of the original system \eqref{eq_ex1_00} and to change its qualitatively.
Therefore, $\rho(x,t)$ is called the \textit{density function}.
Let us consider some cases of how this function can influence the behaviour of the system \eqref{eq_ex1_0}.


\textit{Case 1.1.} 
The density function $\rho(x,t)=\alpha >0$ holds the equilibrium $x=0$ under $d=0$ and takes the same positive value in $D_{+} = \mathbb R$. 
Thus,  $\rho(x,t)$ does not qualitatively affect the exponential stability of the initial system \eqref{eq_ex1_00}, excepting the rate of convergence of the solution of \eqref{eq_ex1_0} to a new domain $D_{S}=\left\{ x \in \mathbb R: |x| > \frac{\bar{d}}{\alpha (1-\mu)} \right\}$, which can be obtained by using Lyapunov function \eqref{eq_ex1_00_Lyap}. 
Thus, increasing the space density $\alpha$ one increases the rate of convergence of \eqref{eq_ex1_0} and decreases the value of ultimate bound, and vice versa, decreasing the space density $\alpha$ one decreases the rate of convergence and increases the value of ultimate bound. 
This fact is shown in the left hand side Fig. \ref{Fig1}, where $d(t)=0.5\textup{atan}(100\sin(0.8t))$, dotted lines are drawn for $\alpha=1$ (initial system \eqref{eq_ex1_00}) and solid lines are drawn for $\alpha=10$. 


\textit{Case 1.2.}  
The density function $\rho(x,t)=\frac{\alpha}{b(t)-|x(t)|}$ with a continuous function $b(t)>0$ preserves a unique equilibrium position $x=0$ under $d(t)=0$. 
We have $\rho(x,t)>0$ in $D_{+}=\{x \in \mathbb R: -b<x<b\}$. 
Also, $\rho(x,t) \to +\infty$ when $|x-b| \to 0$ in $D_+$.
Choosing the quadratic function \eqref{eq_ex1_00_Lyap}, we obtain 
$\dot{V} \leq -\frac{\alpha \mu}{b-|x|} x^2-\frac{\alpha(1-\mu)}{b-|x|} x^2+\bar{d}|x| \leq -\frac{\alpha \mu}{b-|x|} x^2<0$ in  
$x \in D_S = \left\{x \in \mathbb R: x \in \left(-b,-\frac{\bar{d}b}{\bar{d}+\alpha(1-\mu)} \right) \cup \left(\frac{\bar{d}b}{\bar{d}+\alpha(1-\mu)}, b\right) \right\}$. 
The right hand side Fig. \ref{Fig1} shows the transients for $\alpha=1$, $b(t)=2e^{-t}+0.1$, and $d(t)=0.5\textup{atan}(100\sin(0.8t))$. 

\begin{figure}[h!]
\begin{minipage}[h]{0.49\linewidth}
\center{\includegraphics[width=0.9\linewidth]{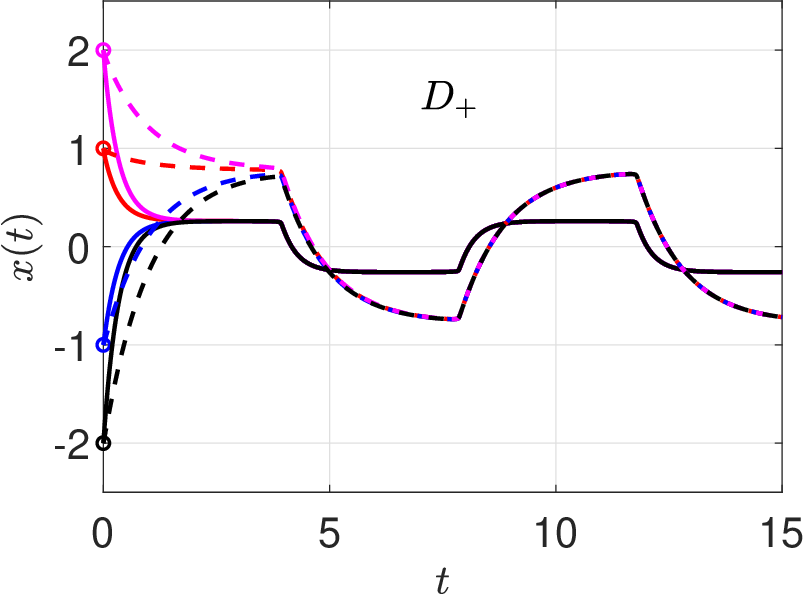}}
\end{minipage}
\hfill
\begin{minipage}[h]{0.49\linewidth}
\center{\includegraphics[width=0.9\linewidth]{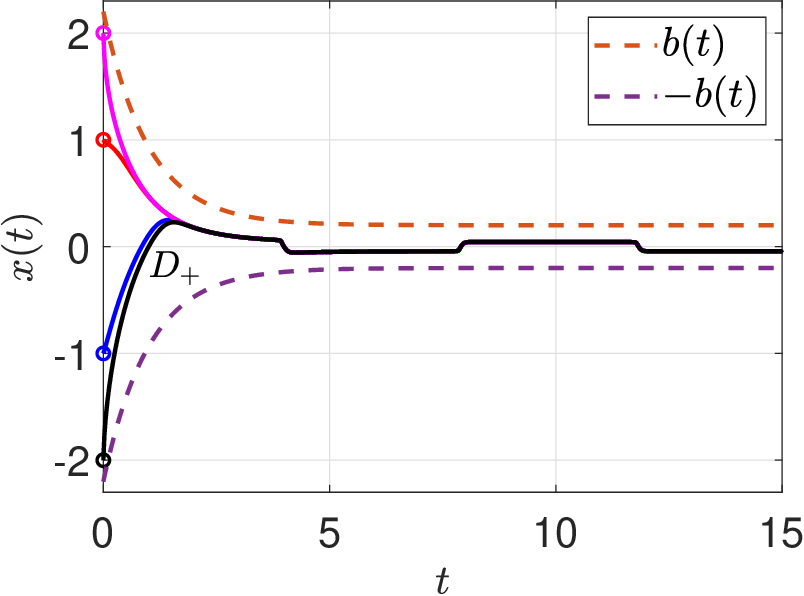}}
\end{minipage}
\caption{The transients for $\rho(x,t)=1$ (dotted lines) and $\rho(x,t)=10$ (solid lines) in \eqref{eq_ex1_0} (left), as well as 
for $\rho(x,t)=\frac{1}{2e^{-t}+0.1-|x(t)|}$ in \eqref{eq_ex1_0} (right) under $d(t)=0.5\textup{atan}(100\sin(0.8t))$.}
\label{Fig1}
\end{figure}


\textit{Case 1.3.}  
Consider the density function $\rho(x,t)=\alpha [x(t)-z(t)] \textup{sign}(x)$ with a continuous function $z(t)$. 
In this case the right hand side of \eqref{eq_ex1_0} is equal to $0$ under $d(t)=0$ when $x=0$ and $x(t)=z(t)$.
The function $\rho(x,t)$ takes a positive value in  
$D_{+}=\{x \in \mathbb R: x \in (-\infty;0) \cup (z;+\infty), ~ z>0\} \cup \{x \in \mathbb R: x \in (-\infty;z) \cup (0;+\infty), z<0\}$ and negative value in $D_{-}=\{x \in \mathbb R: x \in (0;w), z>0\} \cup \{x \in \mathbb R: x \in (z;0), z<0\}$.
Choosing the quadratic function \eqref{eq_ex1_00_Lyap}, we get 
$\dot{V}=-\alpha \mu [x-z] \textup{sign}\{x\}x^2 -\alpha (1-\mu) [x-z] \textup{sign}(x) x^2 +xd$. 
Find the upper estimate as $\dot{V} \leq -\alpha \mu [x-z] \textup{sign}(x)x^2< 0$ in 
$D_S=\Big\{x \in \mathbb R: x \in \left(-\infty;-\sqrt{\frac{\bar{d}}{(1-\mu)\alpha}}\right) \cup \left(z+\sqrt{\frac{\bar{d}}{(1-\mu)\alpha}};+\infty \right)\ , z>0\Big\} \cup \Big\{x \in \mathbb R: x \in \left(-\infty;z-\sqrt{\frac{\bar{d}}{(1-\mu)\alpha}} \right) \cup \left(\sqrt{\frac{\bar{d}}{(1-\mu)\alpha}};+\infty \right), z<0\Big\}$. 
The lower estimate can be written as 
$\dot{V} \geq - \alpha \mu [x-z] \textup{sign}\{x\}x^2 > 0$ in 
$D_{U} = \Big\{ x \in \mathbb R: x \in \left(\sqrt{\frac{\bar{d}}{(1-\mu)\alpha}};z \right), z>0 \Big\} \cup \Big \{x \in \mathbb R: x \in \left(z+\sqrt{\frac{\bar{d}}{(1-\mu)\alpha}};-\sqrt{\frac{\bar{d}}{(1-\mu)\alpha}} \right), z<0 \Big\}$. 
The left hand side Fig. \ref{Fig2} shows the transients for $\alpha=1$, $z(t)=2e^{-t}+0.1$, and $d(t)=0.5\textup{atan}(100\sin(0.8t))$. 


\textit{Case 1.4.}  
Consider the density function $\rho(x,t)= - \alpha \ln \frac{\overline{b}(t)-x(t)}{x(t)-\underline{b}(t) }$ with continuous functions 
$\overline{b}(t)>\underline{b}(t)>0$. 
The right hand side of \eqref{eq_ex1_0} is equal to $0$ when $x(t)=z(t)=0.5[\overline {b}(t)+\underline{b}(t)]$ under $d(t)=0$.
The function $\rho(x,t)$ takes a positive value in 
$D_{+}=\{x \in \mathbb R: z < x < \overline{b} \}$ and a negative value in  
$D_{-}=\{x \in \mathbb R: \underline{b}<x <z \}$.
The trajectories started in $D_{+} \cup D_{-}$ never leave it because $|\rho(x,t)| \to +\infty$ as $x$ approaches $\underline{b}$ and $\overline{b}$.
Choosing the quadratic function \eqref{eq_ex1_00_Lyap}, we obtain $\dot{V}=\alpha \mu \ln \frac{\overline{b}-x}{x-\underline{b}} x^2 + \alpha (1-\mu) \ln \frac{\overline{b}-x}{x-\underline{b}} x^2 + xd$. 
Estimating 
$\dot{V} \leq \alpha \ln \frac{\overline{b}-x}{x-\underline{b}} \mu x^2 -\alpha |x| (1-\mu) \left(\frac{4}{\overline{b}-\underline{b}}(x-z)x - \frac{\bar{d}}{(1-\mu)\alpha} \right) \leq \alpha \ln \frac{\overline{b}-x}{x-\underline{b}} \mu x^2 < 0$, 
one concludes that the system \eqref{eq_ex1_0} is stable in 
$D_S=\left\{x \in \mathbb R: \sqrt{\frac{(\overline{b}-\underline{b})\bar{d}}{4(1-\mu)\alpha}} + z < x < \overline{b} \right\}$.
Estimating $\dot{V} \geq \alpha \ln \frac{\overline{b}-x}{x-\underline{b}} \mu x^2 + \alpha |x| (1-\mu) \left(\frac{4}{\overline{b}-\underline{b}}(x-z)x - \frac{\bar{d}}{(1-\mu)\alpha} \right) \geq \alpha \ln \frac{\overline{b}-x}{x-\underline{b}} \mu x^2 > 0$, 
one gets that the system \eqref{eq_ex1_0} is unstable in 
$D_U=\left\{x \in \mathbb R: \underline{w}<x <\sqrt{\frac{(\overline{b}-\underline{b})\bar{d}}{4(1-\mu)\alpha}}+z \right\}$. 
Fig. \ref{Fig2} (right) shows the transients for $\alpha=1$, $\underline{b}(t)=e^{-0.1t}(2+\sin(t))+0.2$, $\overline{b}(t)=e^{-0.1t}(2+\sin(t))+0.3$, and $d(t)=0.5\textup{atan}(100\sin(0.8t))$.

\begin{figure}[h!]
\begin{minipage}[h]{0.47\linewidth}
\center{\includegraphics[width=0.9\linewidth]{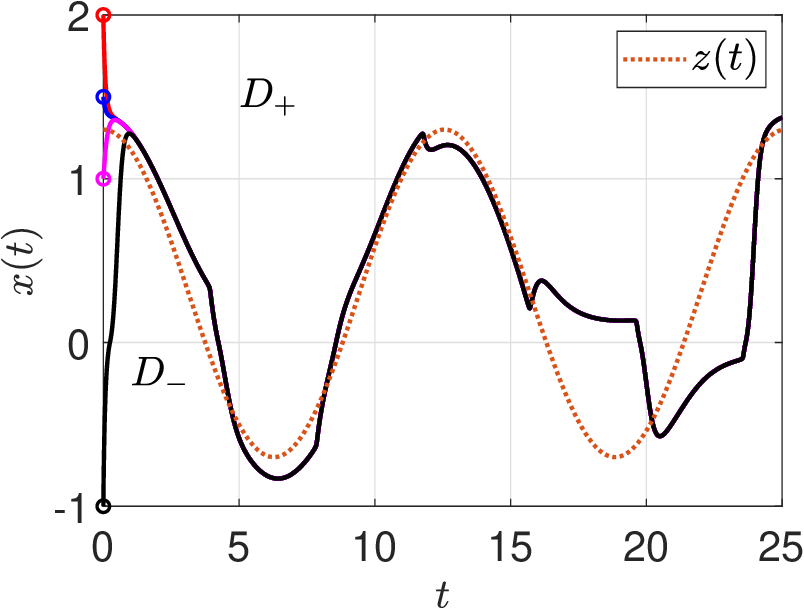}}
\end{minipage}
\hfill
\begin{minipage}[h]{0.47\linewidth}
\center{\includegraphics[width=0.9\linewidth]{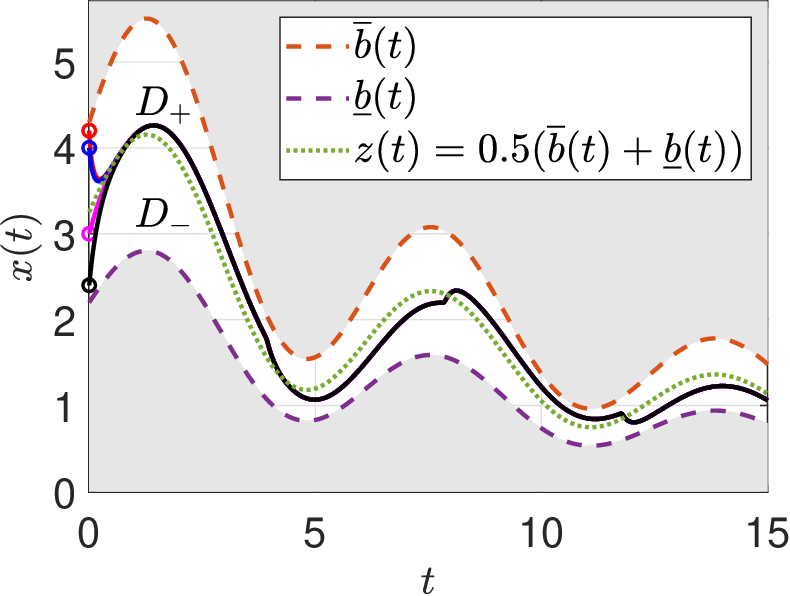}}
\end{minipage}
\caption{The transients for $\rho(x,t)=[x-0.3-\cos(0.5t)] \textup{sign}(x)$ (left) and $\rho (x,t)=\ln \frac{2e^{-0.1t}(2+\sin(t))+0.3-x}{x-e^{-0.1t}(2+\sin(t))-0.2}$ (right) under $d(t)=0.5\textup{atan}(100\sin(0.8t))$.}
\label{Fig2}
\end{figure}

\begin{remark}
One can also consider \eqref{eq1} with piecewise continuous right hand side in $t$. 
However, this will complicate the stability analysis in the sense of the choice of initial conditions \cite{Filippov88}.
Let us demonstrate this in example 1, case 2.1, where consider piecewise continuous functions $w(t)$. 
It can happen that solutions started in the region $D_S$ will leave this region. 
For example, let $w(t) = 2$ for $t \in [0, t_0]$ and $w(t) = 1$ for $t > t_0$. 
If $t_0$ is small enough and the initial condition $x(0)$ is close to $2$ or $-2$, then at $t \to t_0$ the trajectory $x(t)$ will ``hit the wall'', formed by the function $w(t)$, and jump out of $D_S$.
The existence theorem guarantees that the solution $x(t)$ can be extended further, but stability analysis is no longer applicable in this case \cite{Filippov88}. 
This solution will jump out of the $D_S$ region and begin to grow indefinitely. 

Since the present paper is devoted to the study of density systems, for the sake of simplicity and without loss of generality, we consider continuous in $t$ systems in order to simplify the stability analysis and to avoid finding the attraction set.
\end{remark}


Example 1 shows how the density function $\rho(x,t)$ given in the space $(x,t)$ can qualitatively influence the transients of the original system \eqref{eq_ex1_00}. 
Now, in contrast to Example 1 and \cite{Furtat20a,Furtat20b,Furtat21,Furtat22} we will show that the density function can be represented explicitly or implicitly on the right hand side of the system.


\textit{Example 2.}
Consider the system
\begin{equation}
\label{eq_ex2_1}
\begin{array}{l}
\dot{x}_1=x_2-\rho_1(x,t) x_1+d_1,
\\
\dot{x}_2=-x_1-\rho_2(x,t) x_2+d_2,
\end{array}
\end{equation}
where $\rho_1(x,t)$ and $\rho_2(x,t)$ are continuous functions in $t$ and piecewise continuous in $x$ in $\mathbb R^2 \times [0, + \infty)$, 
$d_1(t)$ and $d_2(t)$ are bounded functions with $|d_1(t)|<\bar{d}_1$ and $|d_2(t)|<\bar{d}_2$.
Consider the quadratic function
\begin{equation}
\label{eq_ex2_2}
\begin{array}{l}
V=0.5(x_1^2+x_2^2).
\end{array}
\end{equation}
Taking the derivative of \eqref{eq_ex2_2} along the solutions of \eqref{eq_ex2_1}, one gets
\begin{equation}
\label{eq_ex2_3}
\begin{array}{l}
\dot{V} = - \rho_1 x_1^2 - \rho_2 x_2^2+x_1d_1+x_2d_2.
\end{array}
\end{equation}


\textit{Case 2.1.} 
Let $\rho_1(x,t)=\rho_2(x,t)=\rho(x,t)$. 
Estimate \eqref{eq_ex2_2} as follows
\begin{equation}
\label{eq_ex2_3_est}
\begin{array}{lll}
\dot{V} \leq &  -\rho \mu (x_1^2+x_2^2) - (1-\mu)|x_1| (\rho |x_1| -\frac{\bar{d}}{1-\mu}) 
\\
&- (1-\mu)|x_2| (\rho |x_2| -\frac{\bar{d}}{1-\mu}),
\\
\dot{V} \geq & -\rho \mu (x_1^2+x_2^2) + (1-\mu)|x_1| (\rho |x_1| -\frac{\bar{d}}{1-\mu}) 
\\
&+ (1-\mu)|x_2| (\rho |x_2| -\frac{\bar{d}}{1-\mu}).
\end{array}
\end{equation}

\textit{Case 2.1.1.} Consider $\rho(x,t)=\ln \frac{\overline{b}-|x_1|^{\beta}-|x_2|^{\beta}}{|x_1|^{\beta}+| x_2|^{\beta}-\underline{b}}$, where $\beta>0$, $\overline{b}(t)>\underline{b}(t)>0$ are continuous functions. 
The density function $\rho(x,t)>0$ in $D_{+}=\{x \in \mathbb R^2: \underline{b}<|x_1|^{\beta}+|x_2|^{\beta}<\overline{b}\}$ 
and 
$\rho(x,t)<0$ in $D_{-}=\{x \in \mathbb R^2: \overline{b}<|x_1|^{\beta}+|x_2|^{\beta}< \underline{b}\}$.
Thus, we have $\dot{V} < 0$
in
$D_S=\left\{ x \in \mathbb D_{+}, t \geq 0: \rho|x_1|>\frac{\bar{d}_1}{1-\mu}~\mbox{and}~\rho|x_2|>\frac{\bar{d}_2}{1-\mu} \right\}$ 
and $\dot{V} > 0$ in 
$D_U=\left\{ x \in \mathbb D_{-}: \rho(x,t)|x_1|>\frac{\bar{d}_1}{1-\mu}~\mbox{and}~\rho(x,t)|x_2|>\frac{\bar{d}_2}{1-\mu} \right\}$.

In this case, the density function $\rho(x,t)$ is explicitly represented in the system \eqref{eq_ex2_1},
but it is not multiplied by the whole right hand side, as in example 1 and in \cite{Furtat20a,Furtat20b,Furtat21,Furtat22}.
Fig.~\ref{Fig4} shows phase trajectories for $\beta=1$, $\overline{g}=3$, $\underline{b}=2$ (left) and $\beta=0.6$, $\overline{b}=3^{0.6}$, $\underline{b}=1$ (right) with $x(0)=col\{0,2.5\}$.

\textit{Case 2.1.2.} 
Let the density function be represented by the constant $\rho(x,t)=1$ and the logarithmic function $\rho(x,t)=\sum_{i=1}^4\ln (|x_1-x_{1i}|^{\beta_i}+|x_2-x_{2i}|^{\beta_i}-b_i)$, $\beta_1=0.5$, $\beta_2=1$, $\beta_3=2$, and $\beta_4=4$. 
The first density function $\rho(x,t)>0$ in $D_{+}=\mathbb R^2$. 
The second density function $\rho(x,t)>0$ in $D_{+}=\{x \in \mathbb R^2: |x_1-x_{1i}|^{\beta_i}+|x_2-x_{2i}|^{\beta_i}-b_i>1, i=1,...,4\}$ 
and 
$\rho(x,t)<0$ in $D_{-}=\{x \in \mathbb R^2: 0<|x_1-x_{1i}|^{\beta_i}+|x_2-x_{2i}|^{\beta_i}-b_i<1, i=1,...,4\}$. 
Thus, we have the similar expressions for $D_S$ and $D_U$ as in \textit{case 2.1.1}.

Fig.~\ref{Fig4a} shows the phase trajectories of \eqref{eq_ex2_1} with constant density ($\rho_1(x,t)=\rho_2(x,t)=1$, dashed curves) and logarithmic density (solid line). 
It can be seen that the logarithmic density, unlike the constant one, sets restrictions in the phase space. 
Therefore, the trajectories of the system \eqref{eq_ex2_1} with constant density tend to the neighbourhood of zero through forbidden (grey) region, while the trajectories of the system \eqref{eq_ex2_1} with the logarithmic density go around the forbidden (grey) regions to the neighbourhood of zero.


Here and below:
\begin{itemize}
\item grey areas in the figures mean that the density function is chosen such that there are no solutions of the system in these areas (on the boundary of this area, the density value increases to infinity);
\item the dotted curve indicates the equilibrium position and, accordingly, the boundary between the stable $D_S$ and the unstable $D_U$ regions.
\end{itemize}




\textit{Case 2.2.} 
Let $\rho_1=\alpha \ln(|x_1|^{\beta}+|x_2|^{\beta}-1)$, $\alpha>0$, $\beta>0$, $\rho_2=0$ and $d_2=0$.
Then $\rho(x,t)=\alpha \ln(|x_1|^{\beta}+|x_2|^{\beta}-1)>0$ in $D_{+}=\{x \in \mathbb R ^2:|x_1|^{\beta}+|x_2|^{\beta}>1\}$ and
$\rho(x,t)<0$ in $D_{-}=\{x \in \mathbb R^2: 0<|x_1|^{\beta}+|x_2|^{\beta}<1\}$. 
One has $\dot{V} = -\rho(x,t) \mu x_1^2 - \rho(x,t) (1-\mu) x_1^2 + x_1 d_1$. 
Then, $\dot{V} \leq -\rho(x,t) \mu x_1^2 <0$ in $D_{S}=\{x \in D_{+}, t \geq 0: \rho |x_1|>\frac{\bar{d}}{1-\mu}\}$ and
$\dot{V} \geq -\rho(x,t) \mu x_1^2 >0$ in $D_{S}=\{x \in D_{-}, t \geq 0: \rho |x_1|>\frac{\bar{d}}{1-\mu}\}$. 
Fig.~\ref{Fig6} shows the phase trajectories for $\alpha=5$, $\beta=1$ (left) and $\alpha=5$, $\beta=0.5$ (right).
Unlike \textit{case 2.1} in \textit{case 2.2} the density function $\rho(x,t)$ is only represented in one equation in \eqref{eq_ex2_1}.

\begin{figure}[h]
\begin{minipage}[h]{0.49\linewidth}
\center{\includegraphics[width=0.9\linewidth]{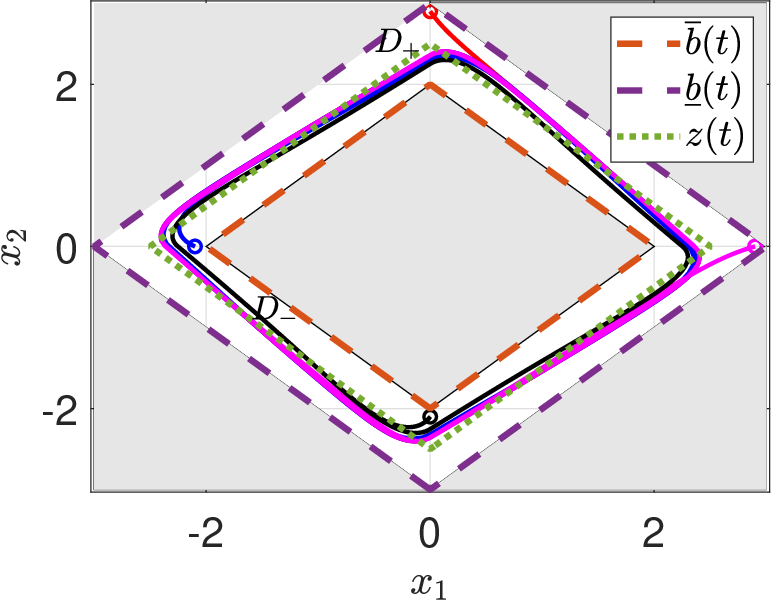}}
\end{minipage}
\hfill
\begin{minipage}[h]{0.49\linewidth}
\center{\includegraphics[width=0.9\linewidth]{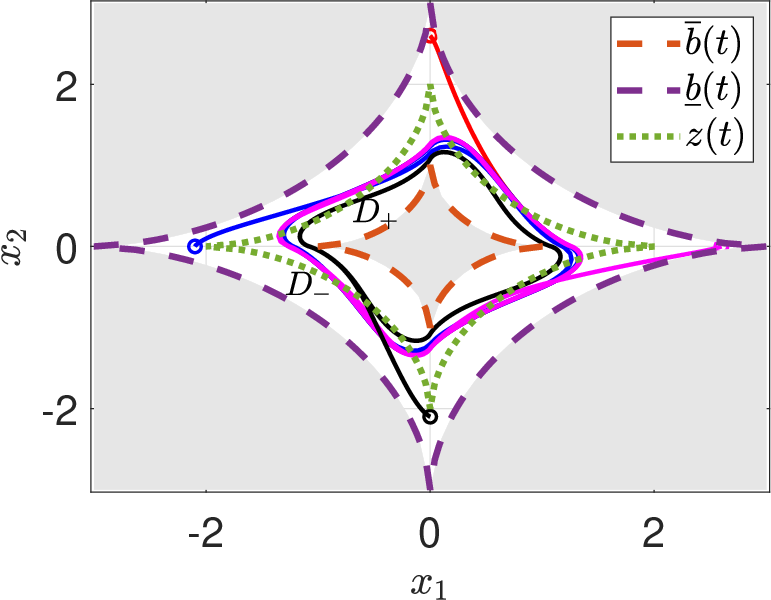}}
\end{minipage}
\caption{The phase trajectory of \eqref{eq_ex2_1} for $\rho(x,t)=\ln \frac{3-|x_1|-|x_2|}{|x_1|+|x_2|-2}$ (left) and $\rho(x,t)=\ln \frac{3^{0.6}-|x_1|^{0.6}-|x_2|^{0.6}}{|x_1|^{0.6}+| x_2|^{0.6}-1}$ (right).}
\label{Fig4}
\end{figure}

\begin{figure}[h]
\center{\includegraphics[width=0.5\linewidth]{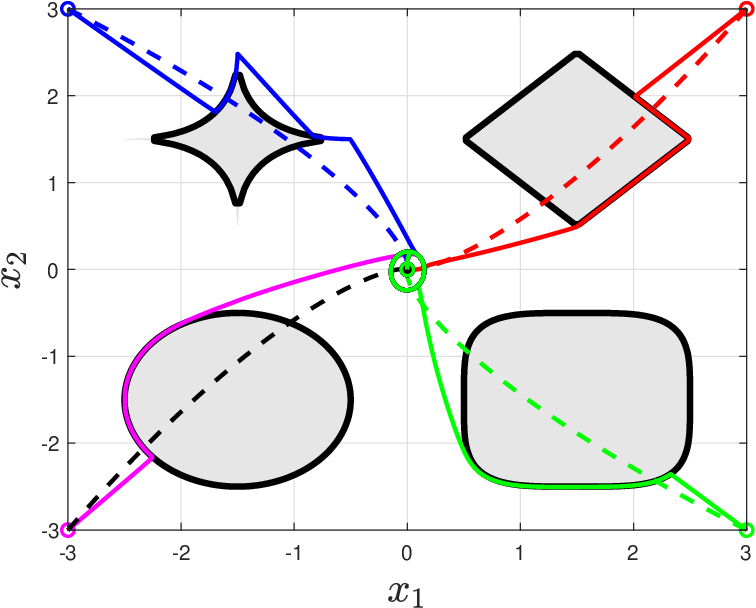}} 
\caption{The phase trajectories of \eqref{eq_ex2_1} for $\rho(x,t)=1$ (dotted curve) and $\rho(x,t)=\sum_{i=1}^4\ln (|x_1-x_{1i}|^{\beta_i}+|x_2-x_{2i}|^{\beta_i}-b_i)$ (solid curve).}
\label{Fig4a}
\end{figure}


\begin{figure}[h]
\begin{minipage}[h]{0.49\linewidth}
\center{\includegraphics[width=0.9\linewidth]{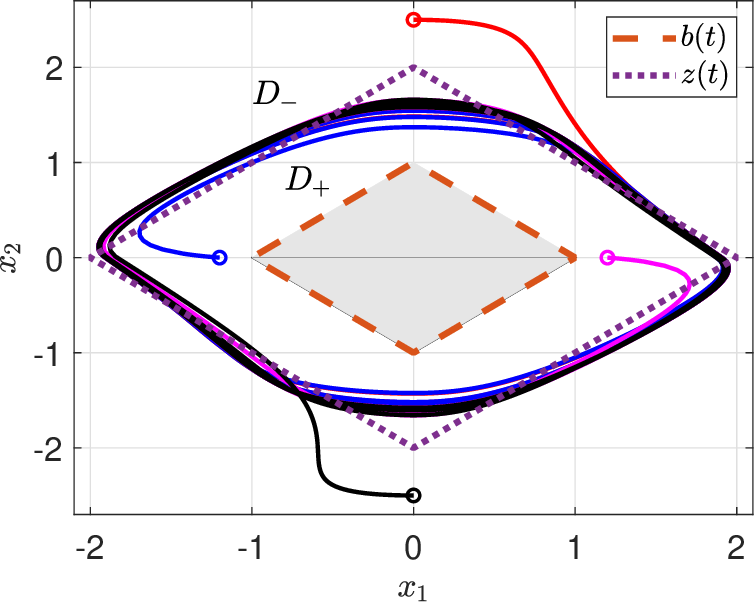}}
\end{minipage}
\hfill
\begin{minipage}[h]{0.49\linewidth}
\center{\includegraphics[width=0.9\linewidth]{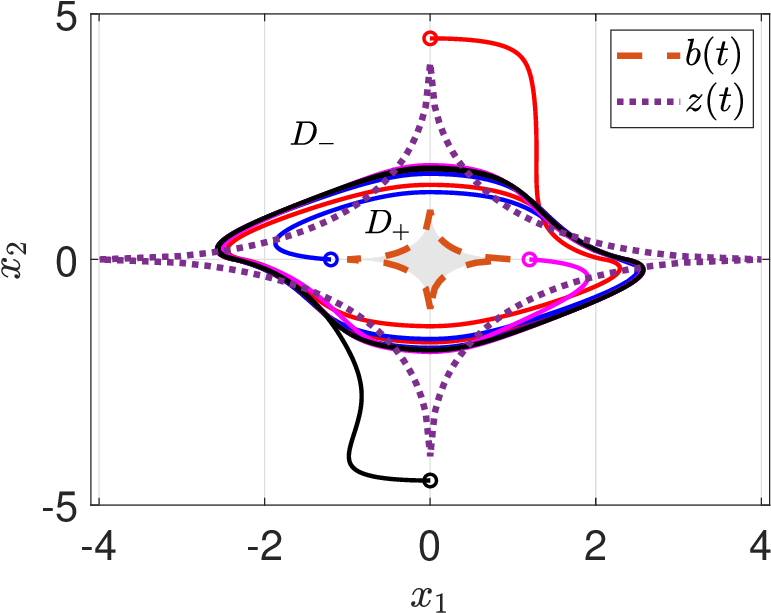}}
\end{minipage}
\caption{The phase trajectories of \eqref{eq_ex2_1} for $\rho(x,t)=5 \ln(|x_1|+|x_2|-1)$ (left) and $\rho(x,t)=5 \ln(|x_1|^{0.5}+|x_2|^{0.5}-1)$ (right).}
\label{Fig6}
\end{figure}


The examples have shown that there may be certain regions where the trajectories of the system tend to the area between these regions. 
We formulate this in the next proposition.


\begin{proposition}
\label{Th1}
Let the system \eqref{eq1} be strictly density in the domains $D_S$ and $D_U$.
If the condition $V(x_{s},t)-V(x_{u},t)>0$ is satisfied for any $t$,
where $x_{s} \in D_S$ and $x_{u} \in D_U$, then the trajectories are attracted to the region between $D_S$ and $D_U$.
If the system \eqref{eq1} satisfies the condition $V(x_{s},t)-V(x_{u},t)<0$ for any $t$,
then the trajectories of the system move away from the region between $D_S$ and $D_U$.
\end{proposition}


\begin{proof}
Let the condition $V(x_{s},t)-V(x_{u},t)>0$ be satisfied for any $t \geq 0$, where $x_{s} \in D_S$ and $x_ {u} \in D_U$.
Since the system is strictly density, by Definition \ref{def3} the conditions $\dot{V} \leq \rho(x,t) W_1(x) < 0$ and $\dot{V} \geq \rho(x,t) W_2(x) > 0$ are satisfied in $D_S $ and $D_U$ respectively.
Hence, the boundary between the regions $D_S$ and $D_U$ is a set to which the trajectories of the system are attracted.

Let the condition $V(x_{s},t)-V(x_{u},t)<0$ be satisfied for each $t \geq 0$, where $x_{s} \in D_S$ and $x_{u } \in D_U$.
By Definition \ref{def3}, $\dot{V} \leq \rho(x,t) W_1(x) < 0$ is satisfied in $D_S$, and $\dot{V } \geq - \rho(x,t) W_2(x) > 0$ in $D_U$.
This means that the separation boundary of the regions $D_S$ and $D_U$ is the set that the trajectories of the system leave.
\end{proof}


\begin{remark}
\label{Rem1}
In Proposition \ref{Th1}, under the attraction of trajectories to a given set we can consider the cases where trajectories approaching a given set over time or finding trajectories in some neighbourhood of a given set.
Moreover, the size of this neighbourhood can remain constant or increase over time.
It depends on the density of the space.
Here are some extreme cases:
\begin{itemize}
\item if in the neighbourhood of the boundary between the regions $D_S$ and $D_U$ the density value decreases to zero, then the trajectories of the system do not approach this boundary.
They can be located both in the neighbourhood of a given boundary as well as moving away from it;
\item if in the neighbourhood of the boundary between the regions $D_S$ and $D_U$ the density neighbourhood increases indefinitely, then the trajectories of the system will approach this boundary.
\end{itemize}
\end{remark}

As we study the density systems, we will highlight special areas.
We have already considered them earlier as grey areas in the figures, now we will define them.

\begin{definition}
\label{def4}
If $\dot{V} \leq \rho(x,t) W_1(x) < 0$ in a neighbourhood of $D_{bh} \times [0, +\infty)$,
there are no solutions in $D_{bh} \times [0, +\infty)$ of \eqref{eq1}
and the value of the density function increases to infinity when approaching this area,
then the domain $D_{bh}$ is said to be absolutely stable.
\end{definition}

\begin{definition}
\label{def5}
If $\dot{V} \geq \rho(x,t) W_2(x) > 0$ in a neighbourhood of $D_{wh} \times [0, +\infty)$,
there are no solutions in $D_{wh} \times [0, +\infty)$ of \eqref{eq1}
and the value of the density function increases to infinity when approaching this area,
then the domain $D_{wh}$ is said to be absolutely unstable.
\end{definition}

\textit{Example 3.}
Consider the system \eqref{eq_ex2_1}, where $\rho_1(x,t)=\rho_2(x,t)=\rho(x)$.
Choose the quadratic function \eqref{eq_ex2_2}.
Then $\dot{V} = - \rho(x) (x_1^2+x_2^2)$.

\textit{Case 3.1.} 
Let $\rho(x)=e^{(x_1^2+x_2^2-1)^{-0.98}}$. 
The density function is positive in $D_{+}=\mathbb R^2 \setminus D_{bh}$, where $D_{bh}=\{x \in \mathbb R^2: x_1^2+x_2^2 \leq 1\}$. 
Furthermore, the density value increases when approaching the region $D_{bh}$. 
The Fig. ~\ref{Fig_bh} (left) shows the phase trajectories for different initial conditions and disturbances $d_1(t)=4 \sin(t)e^{-0.09t}$, $d_2(t)=4 \cos(t)e^{-0.1t}$. 
It can be seen that from some initial conditions the trajectories tend to the region $D_{bh}$ despite the presence of disturbances.
However, for other initial conditions, external disturbances moving away the trajectories from $D_{bh}$ until the influence of disturbances ends, 
after which these trajectories tend to the region $D_{bh}$. 
If there are no disturbances, then $D_{S}=D_{+}$. 

\textit{Case 3.2.} 
Now let us consider $\rho(x)=-\ln(x_1^2+x_2^2-1)$. 
Then the density function is positive in the domain $D_{+}=\{x \in \mathbb R^n: x_1^2+x_2^2 \geq \sqrt{2} \}$ and negative in the region $D_{+}=\{x \in \mathbb R^n: 1 \leq x_1^2+x_2^2 \leq \sqrt{2}\}$. 
Moreover, the density value increases when approaching from the boundary between $D_{+}$ and $D_{-}$ to the region $D_{bh}=\{x \in \mathbb R^2: x_1^2+x_2^2 \leq 1\}$. 
In Fig. ~\ref{Fig_bh} (right) one can see the phase trajectories under $d_1(t)=\sin(t)$ and $d_2(t)=0.8\cos(t)$. 
From some initial conditions in $D_{+}$ there are trajectories that tend to the region $D_{bh}$ despite the presence of perturbations. 
However, there are initial conditions in $D_{+}$ from which the trajectories move into the region $D_{-}$ due to perturbations and stay there. 
If there are no disturbances, then $D_{S}=D_{-}$ and $D_{U}=D_{+}$.

\begin{figure}[h]
\begin{minipage}[h]{0.49\linewidth}
\center{\includegraphics[width=0.9\linewidth]{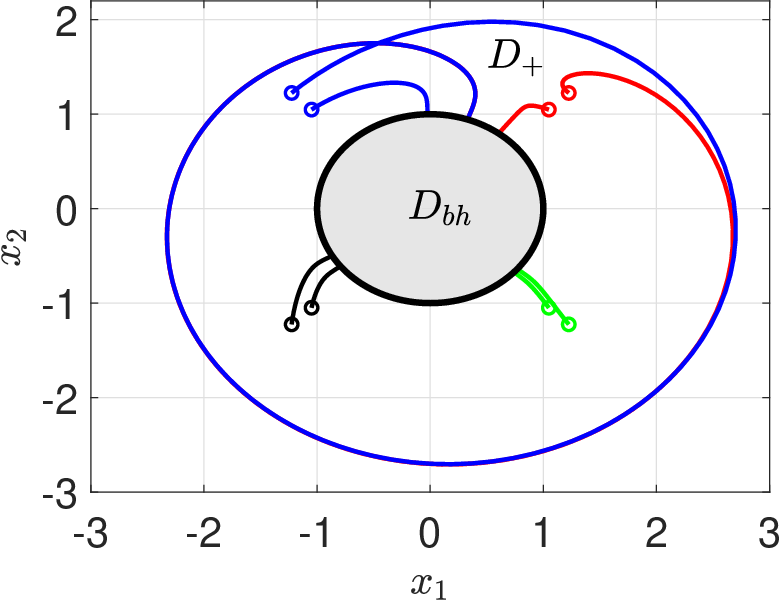}}
\end{minipage}
\hfill
\begin{minipage}[h]{0.49\linewidth}
\center{\includegraphics[width=0.9\linewidth]{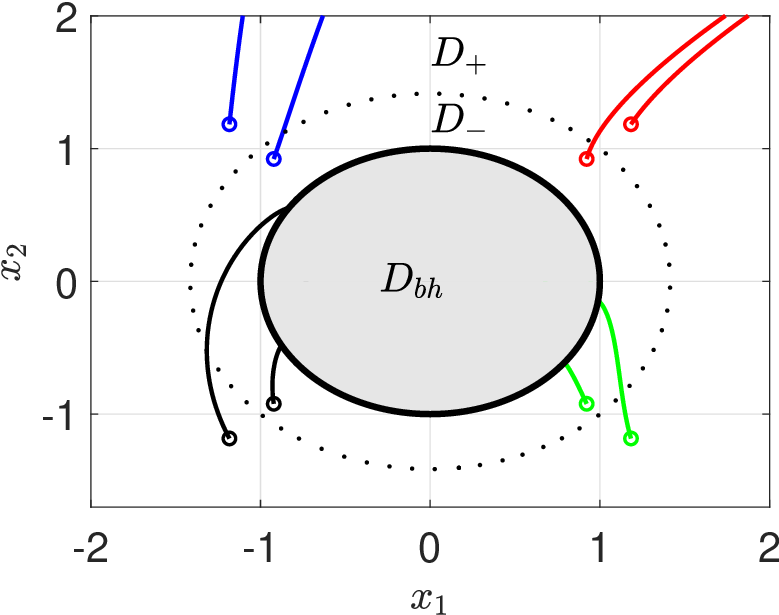}}
\end{minipage}
\caption{The phase portrait of \eqref{eq_ex2_1} for $\rho(x)=e^{(x_1^2+x_2^2-1)^{-0.98}}$ (left) and $\rho(x)=-\ln(x_1^2+x_2^2-1)$ (right).}
\label{Fig_bh}
\end{figure}

\textit{Case 3.3.} 
If $\rho(x)=-e^{(x_1^2+x_2^2-1)^{-0.98}}$, then the density is negative in $D_{-} = \mathbb R^2 \setminus D_{wh}$, where $D_{wh}=\{x \in \mathbb R^2: x_1^2+x_2^2 \leq 1\}$. 
Therefore, despite the presence of disturbances, all trajectories move away from the region $D_{wh}$ (see Fig.~\ref{Fig_wh} on the left) and will never be able to approach the boundary of this region, where the density will increase to infinity. 

\textit{Case 3.4.} 
If $\rho(x)=\ln(x_1^2+x_2^2-1)$, then the density has positive value in $D_{+}=\{x \in \mathbb R^2: x_1^2+x_2^2 \geq \sqrt{2}\}$ and negative value in $D_{-}=\{x \in \mathbb R^2: 1 \leq x_1^2+x_2^2 \leq \sqrt{2}\}$. 
Therefore, the trajectories with initial conditions from $D_{+}$ and $D_{-}$ tend the region of separation of these sets considering disturbances.

\begin{figure}[h]
\begin{minipage}[h]{0.49\linewidth}
\center{\includegraphics[width=0.9\linewidth]{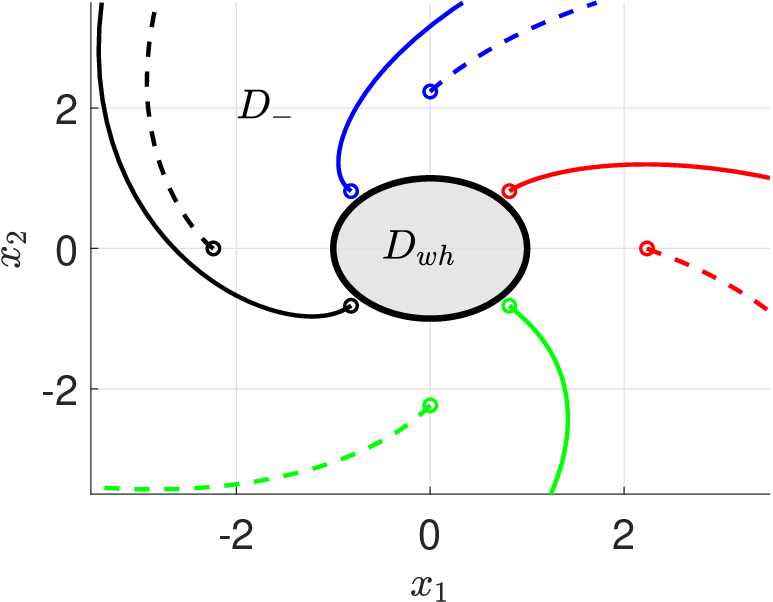}}
\end{minipage}
\hfill
\begin{minipage}[h]{0.49\linewidth}
\center{\includegraphics[width=0.9\linewidth]{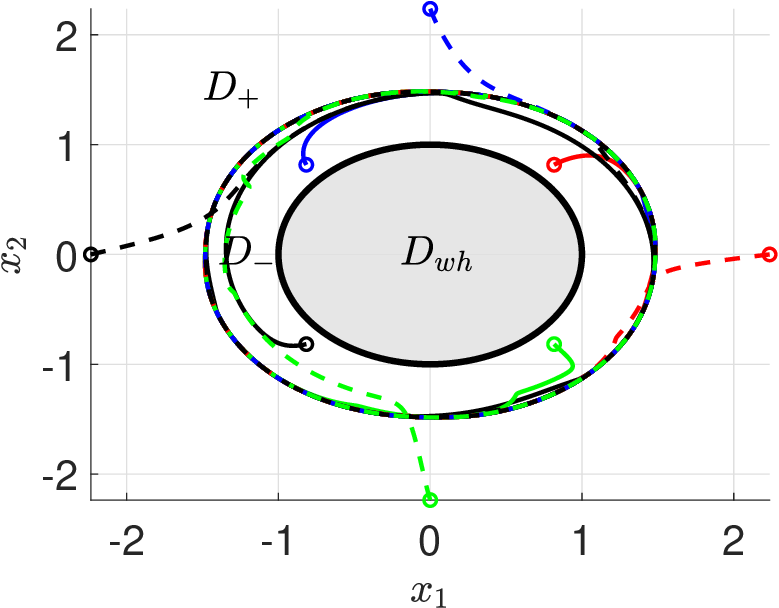}}
\end{minipage}
\caption{Phase portrait of \eqref{eq_ex2_1} for $\rho(x)=-e^{(x_1^2+x_2^2-1)^{-0.98}}$ (left) and
$\rho(x)=\ln(x_1^2+x_2^2-1)$ (right).}
\label{Fig_wh}
\end{figure}


\begin{remark}
\label{physics_mean}
Let us explain the physical meaning of the systems under consideration.
If the density function is explicitly represented on the right hand side of the system equation, for example, in the form $\dot{x}=\rho(x)f(x)$, then the space density value directly affects the phase flow velocity.
Thus, if $\rho(x)=1$, then we have an initial system of the form $\dot{x}=f(x)$.
If $\rho(x)>0$, then the presence of the density function does not qualitatively affect the equilibrium positions and their types, but quantitatively affects the phase portrait.
At $0<\rho(x)<1$ the norm of the phase velocity vector decreases because the space density decreases.
Conversely, when $\rho(x)>1$, the norm of the phase velocity vector increases due to the increase of the space density.
Changing the sign of the density function qualitatively changes the phase portrait.

Condition (b) in Definition \ref{def1} can be interpreted as the rate of change of the phase volume, given by the function $V(x,t)$, taking into account the density of space.

Here are models of real processes:
\begin{itemize}

\item pendulum equation $\dot{x}_1=x_2$, $\dot{x}_2=-\frac{g}{l} \sin x_1 - \frac{k}{m}x_2$, 
where $x_1 $ is an angular deviation of the pendulum from the vertical axis, 
$x_2$ is an angular velocity of the pendulum, 
$g$ is the gravity acceleration, 
$l$ is a pendulum length, 
and $k$ is a friction coefficient \cite{Khalil09}.
Choosing the quadratic function in the form of the total energy as $V=\frac{g}{l} (1-\cos x_1)+0.5x_2^2$, we obtain $\dot{V}=-\frac{k}{m} x_2^2$.
If we set the density function $\rho(x,t)=k$, then we have undamped oscillations in the absence of friction ($\rho(x,t)=0$) and damped oscillations in the presence of friction ($\rho(x,t ) \neq 0$);

\item different types of breeding patterns can be written as $\dot{x}=\rho(x)x$, where $x$ is the size of the biological population \cite{Arnold12}.
For $\rho(x)=k>0$ we have a normal reproduction model, for $\rho(x)=kx$ one has an explosion model, for $\rho(x)=1-x$ we have a logistic curve model;

\item absolutely stable and absolutely unstable regions from Definitions \ref{def4} and \ref{def5} can be found as the simplest models of black and white holes, respectively, see \cite{Carroll04}.
\end{itemize}
\end{remark}

\begin{remark}
According to Remark \ref{physics_mean}, the density systems can be also considered and called as gravitational systems, where $\rho(x,t)$ is a gravitational function. 
Practically, this means that behaviour of systems can be influenced not by the density of space, but by the gravitational field created by a dense body. 
Therefore, the density control in the following sections can be considered as a gravitational control, that depending on the environmental conditions.
However, all this does not affect the mathematical description presented in this paper.
\end{remark}


\section{Density adaptive control}
\label{Sec3}

Consider the system
\begin{equation}
\label{eq1_simple}
\begin{array}{l}
Q(p)y(t)=kR(p)u(t)+d(t),
\end{array}
\end{equation}
where $y \in \mathbb R$ is the output,
$u \in \mathbb R$ is the control, 
$p=d/dt$, 
$Q(\lambda)$ and $R(\lambda)$ are monic polynomials with constant unknown coefficients and orders $n$ and $n-1$ accordingly, 
$R(\lambda)$ is Hurwitz polynomial, 
$\lambda$ is a complex variable, 
$k>$ is a known gain, 
$d \in \mathbb R$ is the disturbance and $\bar{d}=\sup\{|d(t)|\}$. 
Assume that the unknown coefficients of $Q(p)$, $R(p)$, $k$ and initial conditions of \eqref{eq1_simple} belong to a known compact set $\Xi$.

The relative degree of \eqref{eq1_simple} is equal to $1$.
All obtained results can be extended to systems with a relative degree greater than $1$, e.g., by using \cite{Fradkov00,Ioannou12,Tao14,Annaswamy21}.
In this section, we consider systems with a relative degree of $1$ in order to avoid cumbersome derivations to overcame the problem of a high relative degree.

Rewrite the operators $Q(p)$ and $R(p)$ as $Q(p)=Q_m(p)+\Delta Q(p)$ and $R(p)=R_m(p)+\Delta R( p)$,
where $R_m(\lambda)$ is an arbitrary Hurwitz polynomial of order $n-1$. 
The polynomial $Q_m(\lambda)$ of degree $n-1$ is chosen such that $Q_m(\lambda)/R_m(\lambda)=\lambda$. 
The orders of $\Delta Q(p)$ and $\Delta R(p)$ are $n-1$ and $n-2$, respectively.
Extracting the integer part in
$\frac{\Delta Q(\lambda)}{Q_m(\lambda)}=k_{0y}+\frac{\Delta \tilde{Q}(\lambda)}{R_m(\lambda)}$, rewrite \eqref{eq1_simple} as follows
\begin{equation}
\label{eq0_adapt}
\begin{array}{l}
\dot{y}(t)= \left(1+\frac{\Delta R(p)}{R_m(p)} \right)u(t)-\left(\frac{\Delta \tilde{ Q}(p)}{R_m(p)}+k_{0y} \right) y(t)+\hat{d}(t),
\end{array}
\end{equation}
where $\hat{d}(\lambda) = \frac{1}{Q_m(\lambda)}d(\lambda) + \frac{D(\lambda)}{Q_m(\lambda)}$, $D(\lambda)$ is a polynomial depending on the initial conditions of \eqref{eq1_simple} and obtained by Laplace transform. 
Since $Q_m(\lambda)$ is Hurwitz polynomial and $d(t)$ is bounded, then $\hat{d}(t)$ is bounded.

Let $c_0=col\{c_{0y},c_{0u},k_{0y}\}$ be the vector of unknown parameters,
$\Delta \tilde{Q}(p)=c_{0y}^{\rm T}[1~ p~ ... ~p^{n-2}]$ and $\Delta R(p)= c_{0u}^{\rm T}[1~ p~ ... ~p^{n-2}]$.
Also, consider the vector $w=col\{V_y,V_u,y\}$ composed of the following filters
\begin{equation}
\label{eq1_adapt_filt}
\begin{array}{l}
\dot{V}_y(t)=FV_y(t)+by(t),
~~~
\dot{V}_u(t)=FV_u(t)+bu(t).
\end{array}
\end{equation}
Here $F$ is Frobenius matrix with characteristic polynomial $R_m(\lambda)$ and
$b=col\{0,...,0,1\}$.
Thus, the equation \eqref{eq0_adapt} can be rewritten as
\begin{equation}
\label{eq1_adapt}
\begin{array}{l}
\dot{y}(t)=u(t)-c_0^{\rm T}w(t)+\hat{d}(t).
\end{array}
\end{equation}

Introduce the control law
\begin{equation}
\label{eq2_adapt}
\begin{array}{l}
u(t) = c^{\rm T}(t)w(t) + \tau \rho(y,t).
\end{array}
\end{equation}

Substituting \eqref{eq2_adapt} into \eqref{eq1_adapt}, we get the closed-loop system:
\begin{equation}
\label{eq3_adapt}
\begin{array}{l}
\dot{y}(t)=\tau \rho(y,t)+(c(t)-c_0)^{\rm T}w(t)+\hat{d}(t).
\end{array}
\end{equation}


\begin{theorem}
\label{th_adapt} 
Let $\rho(y,t)$ be the function such that $\rho(y,t)y<0$ when $y>g$ and $\rho(y,t)y>0$ when $y<g$, where $g(t)$ is a bounded function. 
Then the control law \eqref{eq2_adapt} together with the adaptation algorithm
\begin{equation}
\label{eq6_adapt}
\begin{array}{l}
\dot{c} = -\beta y w - \gamma c |y|  \textup{sign}(\rho(y,t)y),
\end{array}
\end{equation}
where $\beta, \gamma >0$, leads the system \eqref{eq1_adapt} to the density type and 
$\limsup\limits_{t \to \infty}|y(t)-g(t)|<\delta$.
Also, all signals are bounded in the closed-loop system.
\end{theorem}


\begin{proof}
To analyse the stability of the closed-loop system, introduce Lyapunov function in the form
\begin{equation}
\label{eq4_adapt}
\begin{array}{l}
V=\frac{1}{2}y^2 + \frac{1}{2\alpha}(c-c_0)^{\rm T}(c-c_0).
\end{array}
\end{equation}

Consider the relation
$(c-c_0)^{\rm T}c=-0.5(c-c_0)^{\rm T}(c-c_0)-0.5c^{\rm T}c+0.5c_0^{\rm T}c_0$. 
Find the time derivative of \eqref{eq4_adapt} along the trajectories \eqref{eq3_adapt}, \eqref{eq6_adapt} and rewrite the result as follows
\begin{equation}
\label{eq7_adapt}
\begin{array}{l}
\dot{V}=\tau \mu \rho(y,t) y + \tau (1-\mu) \rho(y,t) y + y \hat{d} 
\\
+ \frac{\gamma}{2\alpha} [-(c-c_0)^{\rm T}(c-c_0)-c^{\rm T}c+c_0^{\rm T}c_0 ] |y| \textup{sign}(\rho y).
\end{array}
\end{equation}

\textit{Case (i).} 
Let $y>g$. Consider the upper estimate of \eqref{eq7_adapt} as follows
\begin{equation}
\label{eq7_adapt_case_i1}
\begin{array}{l}
\dot{V} \leq \tau \mu \rho(y,t) y - |y| \left ((1-\mu)\tau |\rho(y,t)| - \bar{d}- \frac{\gamma }{2\beta} \bar{c}_0^2  \right).
\end{array}
\end{equation}
where $\bar{c}_0$ is the largest value of $|c_0|$ from the compact set $\Xi$.
Denoting by
\begin{equation}
\label{eq7_adapt_case_note}
\begin{array}{l}
\eta=\frac{1}{(1-\mu)\tau} \left( \bar{d} + \frac{\gamma}{2\beta} \bar{c}_0^2 \right),
\end{array}
\end{equation}
we get $\dot{V} \leq \tau \mu \rho(y,t) y < 0$ in the domain 
$D_S=\left\{y \in \mathbb R, t \geq 0: y>g ~\mbox{and}~ |\rho| > \eta \right\}$. 
Thus, the closed-loop system is stable in $D_S$.

\textit{Case (ii).} 
Let $y<g$. Consider the lower estimate of \eqref{eq7_adapt} in the form
\begin{equation}
\label{eq7_adapt_case_i1}
\begin{array}{l}
\dot{V} \geq \tau \mu \rho(y,t) y + |y| \left ((1-\mu)\tau |\rho(y,t)| - \bar{d}- \frac{\gamma}{2\beta} \bar{c}_0^2 \right).
\end{array}
\end{equation}
One gets $\dot{V} \geq \tau \mu \rho(y,t) y > 0$ in the domain 
$D_U=\left\{y \in \mathbb R: y < g ~\mbox{and}~ |\rho(y,t)| > \eta \right\}$. 
Therefore, the closed-loop system is unstable in $D_{U}$.

According to Proposition \ref{Th1}, we have $\limsup\limits_{t \to \infty} |y(t)-g(t)| < \delta$, where $\delta>0$.
The expression \eqref{eq3_adapt} implies that $\limsup\limits_{t \to \infty} \left|(c(t)-c_0)^{\rm T} w(t) \right| < \infty$.
The boundedness of $V_y(t)$ follows from the first equation \eqref{eq1_adapt_filt}, the boundedness of $y(t)$, and the Hurwitz matrix $F$.
Substituting \eqref{eq2_adapt} into the second equation \eqref{eq1_adapt_filt}, we get
\begin{equation}
\label{eq8_adapt_filt_proof}
\begin{array}{lll}
\dot{V}_u & =FV_u+bc_0^{\rm T}w+b(c-c_0)^{\rm T}w+ b \tau \rho (y,t) 
\\
&=(F+bc_{0u})V_u+b[c_{0y}^{\rm T}V_y+k_{0y}y
\\
&+(c-c_0)^{\rm T}w+ \tau \rho(y,t)].
\end{array}
\end{equation}
The matrix $F+bc_{0u}$ has Hurwitz characteristic polynomial $R(\lambda)$.
Therefore, for the bounded terms in square brackets, $V_u(t)$ is ultimately bounded.
Then $w(t)$ is also ultimately bounded.
From the condition $\limsup\limits_{t \to \infty} |y(t)-g(t)| < \delta$, the boundedness of $g(t)$ and the ultimate boundedness of $w(t)$ it follows from \eqref{eq6_adapt} that $\dot{c}(t)$ and $c(t)$ are ultimately bounded.
Then \eqref{eq2_adapt} implies boundedness of the control law.
As a result, all signals in the closed-loop system are bounded.
\end{proof}


\begin{remark}
\label{rem_adapt} 
If $\rho(x,t)$ is chosen such that $|\rho(x,t)| > \eta$ for any $x$ and $t$ in Theorem \ref{th_adapt}, 
then the control law \eqref{eq2_adapt} with the adaptation algorithm \eqref{eq6_adapt} 
ensures $\lim\limits_{t \to \infty}(y(t)-g(t)) = 0$. 
This follows from the proof of Theorem \ref{th_adapt} and Proposition \ref{Th1}.
\end{remark}


\textit{Example 3}.
Consider the system \eqref{eq1_simple} with unknown coefficients of 
$Q(p)=(p-1)^3$, $R(p)=(p+1)^2$, and $k=1$, as well as unknown initial conditions $p^2y(0)=0$ , $py(0)=0$, $y(0)=4$.

Choose 
$F=\begin{bmatrix}
0 & 1\\
-1 & -2
\end{bmatrix}$
in \eqref{eq1_adapt_filt}.
In the adaptation algorithm \eqref{eq6_adapt} set $\beta=0.1$ and $\gamma=0.01$.
Below we consider different types of the density function $\rho(y,t)$ in \eqref{eq2_adapt}.


1) For $\rho(y,t)= - \alpha y$, $\alpha>0$ the closed-loop system \eqref{eq3_adapt} has the equilibrium point $y=0$ under $d(t)=0$. 
The density function $\rho(y,t)$ is negative in $D_{-}=\mathbb R_{+}$ and positive in $D_{+}=\mathbb R_{-}$ for any $t \geq 0$. 
Substituting $\rho(y,t)$ into \eqref{eq7_adapt_case_i1}, one gets $\dot{V} < 0$ in 
$D_S=\left\{y \in \mathbb R: |y| > \frac{\eta}{\alpha}\right\}$, where $\eta$ is given by \eqref{eq7_adapt_case_note}.
We have obtained the problem of adaptive stabilization, which is described in detail in \cite{Fradkov00,Ioannou12,Tao14,Annaswamy21,Bobtsov15,Bobtsov22}.
Fig.~\ref{Fig_ad_1} (left) (see only the trajectory entering the grey area) shows the transients for $\alpha=4$ and $d(t)=5+10\sin(7t)$.


2) For $\rho(y,t)= \alpha \ln\frac{b-y}{b+y}$ with continuous function $b(t)>0$ the closed-loop system \eqref{eq3_adapt} has the equilibrium $y=0$ under $d(t)=0$.
The density function is negative in $D_{-}=\{y \in \mathbb R_{+}: 0<y<b \}$ and positive in $D_{+}=\{y \in \mathbb R_{-}: -b<y<0\}$. 
Substituting $\rho(y,t)$ into \eqref{eq7_adapt_case_i1}, we have $\dot{V}<0$ in $D_S = \left\{ y \in \mathbb R: \frac{b(e^\frac{\eta}{\alpha}-1)}{1+e^\frac{\eta}{\alpha}}<|y|<b \right\}$.
Moreover, $\rho(y,t) \to -\infty$ for $|y-b| \to 0$ in $D_-$ and $\rho(y,t) \to +\infty$ for $|y-b| \to 0$ in $D_+$.
We have obtained a stabilization problem with symmetric constraints $-b$ and $b$.
Fig.~\ref{Fig_ad_1} (left) shows the transients (trajectory inside the dotted tube) for $b(t)=\begin{cases}
    1 & t \leq 0.5,\\
    -9t+9.5 & 0.5 < t \leq 1, \\
    0.5, & t >1,
  \end{cases}$ as well as the same $\alpha=4$ and $d(t)=5+10\sin(7t)$ as in the previous case.
In contrast to the classical control scheme (the trajectory corresponding to $\rho(y,t)= - \alpha y$),
using the density function $\rho(y,t)= \alpha \ln\frac{b-y}{b+y}$ guarantees that the transients are in the tube at all times.

\begin{figure}[h]
\begin{minipage}[h]{0.49\linewidth}
\center{\includegraphics[width=0.9\linewidth]{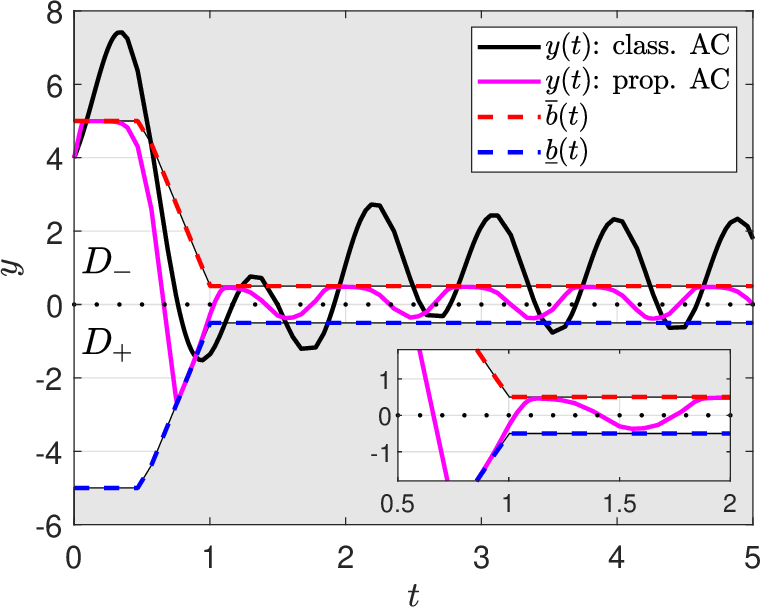}}
\end{minipage}
\hfill
\begin{minipage}[h]{0.49\linewidth}
\center{\includegraphics[width=1\linewidth]{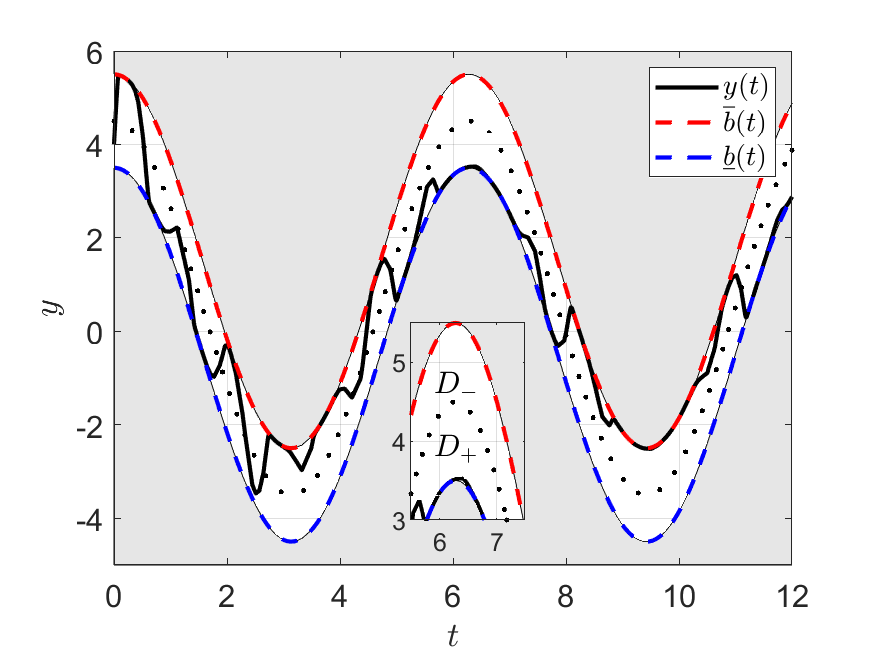}}
\end{minipage}
\caption{In the left figure, the transients in the adaptive control scheme with density functions $\rho(y,t)= - \alpha y$ (the curve crossing a grey area) and $\rho(y,t)= \alpha \ln\frac{b-y}{b+y}$ (the curve inside a tube with dashed borders). In the right figure, the transients in the adaptive control scheme with density function $\rho(y,t)= \alpha \ln\frac{\overline{b}-y}{y-\underline{b}}$.}
\label{Fig_ad_1}
\end{figure}



3) For $\rho(y,t)= \alpha \ln\frac{\overline{b}-y}{y-\underline{b}}$ with continuous functions $\overline{b}>\underline{b}$ one has $\rho(y,t)=0$ in $y =z=0.5(\overline{b}+\underline{b})$ and any $t$.
The density function $\rho(y,t)$ is negative in
$D_{-}=\left\{y \in \mathbb R: z <y < \overline{b} \right\}$ and
positive in
$D_{+}=\left\{y \in \mathbb R: \underline{b} < y < z \right\}$.
Substituting $\rho(y,t)$ into \eqref{eq7_adapt}, we have $\dot{V}<0$ in
$D_S=\left\{y \in \mathbb R_{+}: \frac{\underline{b}e^\frac{\eta}{\alpha}+\overline{b}}{1+e^\frac{\eta}{\alpha}} < y < \overline{b} \right\} \cup
\left\{y \in \mathbb R_{-}: \underline{b}<y<\frac{\underline{b}e^\frac{\eta}{\alpha}+\overline{b}}{1+e^\frac{\eta}{\alpha}}\right\}$ 
and
$\dot{V} >0$ in
$D_U=\left\{y \in \mathbb R_{+}: \underline{b} < y < \frac{\underline{b}e^\frac{\eta}{\alpha}+\overline{b}}{1+e^\frac{\eta}{\alpha}} \right\} \cup 
\left\{y \in \mathbb R_{-}: \frac{\underline{b}e^\frac{\eta}{\alpha}+\overline{b}}{1+e^\frac{\eta}{\alpha}} < y < \overline{b} \right\}$.
Moreover, $\rho(y,t) \to -\infty$ at $y \to \overline{b}-0$ and $\rho(y,t) \to +\infty$ at $y \to \underline{b}+0$.
We have obtained a stabilization problem with asymmetric constraints $\overline{b}$ and $\underline{b}$.
Fig.~\ref{Fig_ad_1} (right) shows the transients for $\alpha=4$, $\overline{b}=4\cos(t)+1.5$, $\underline{b}=4\cos(t)-0.5$, and $d(t)=-2+10\sin(7t)$.



4) For $\rho(y,t)= - \alpha (y-z)$ with a piecewise continuous function $z$ we have $\rho(y,t)=0$ in $y=z$ for any $t$. 
The density function $\rho(y,t)$ is negative in
$D_{-}=\{y \in \mathbb R : y >z \}$ and
positive in
$D_{+}=\{y \in \mathbb R: y <z\}$.
Substituting $\rho(y,t)$ into \eqref{eq7_adapt}, we have $\dot{V}<0$ in 
$D_S=\left\{y \in \mathbb R_{+}: y >z+\frac{\eta}{\alpha}\right\} \cup \left\{y \in \mathbb R_{-}: y >z-\frac{\eta}{\alpha}\right\}$
and $\dot{V}>0$ in 
$D_U=\left\{y \in \mathbb R_{+}: y <z-\frac{\eta}{\alpha}\right\} \cup \left\{y \in \mathbb R_{-}: y <z+\frac{\eta}{\alpha}\right\}$.
We got the problem of tracking $y$ for $z$.
Fig.~\ref{Fig_ad_4} (left) shows the transients for $\alpha=100$, $z=1+\sin(t)+\textup{atan}[10^3\sin(1.3t)]$,
and $d(t)=5+10\sin(3t)$.

\begin{figure}[h]
\begin{minipage}[h]{0.49\linewidth}
\center{\includegraphics[width=0.9\linewidth]{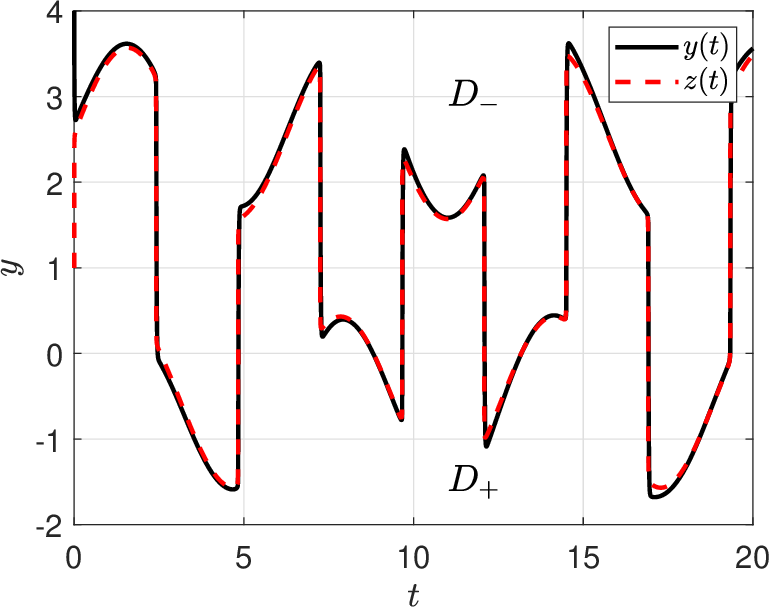}}
\end{minipage}
\hfill
\begin{minipage}[h]{0.49\linewidth}
\center{\includegraphics[width=1\linewidth]{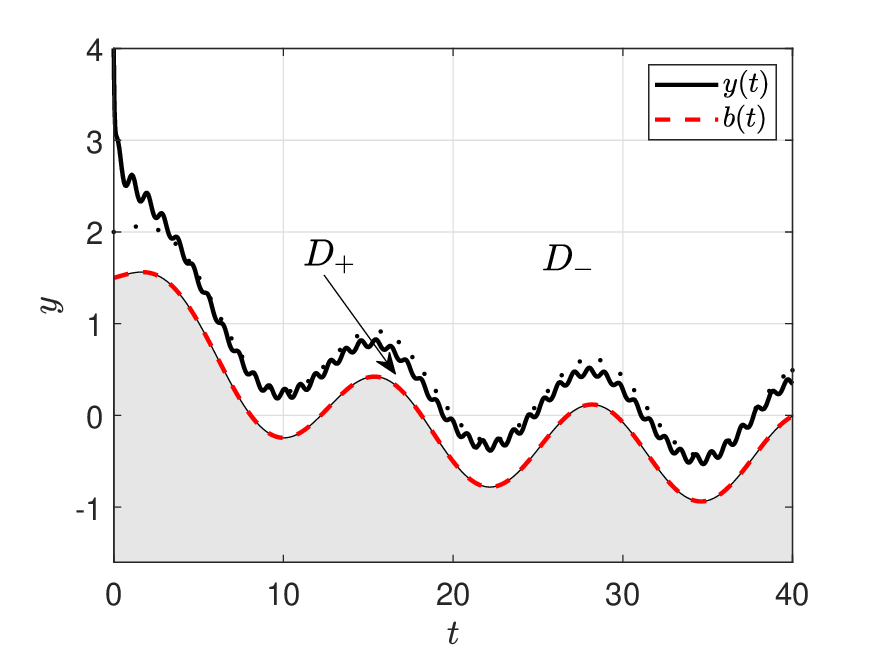}}
\end{minipage}
\caption{In the left figure, the transients in in the adaptive control scheme with density function $\rho(y,t)= - \alpha (y-z)$. In the right figure, the transients in the adaptive control scheme with density function $\rho(y,t)= - \alpha \ln(y-b)$.}
\label{Fig_ad_4}
\end{figure}



5) For $\rho(y,t)= - \alpha \ln(y-b)$ with a continuous function $b$ one has $\rho(y,t)=0$ in $y=z=b+1$ and any $t$.
The density function $\rho(y,t)$ is negative in
$D_{-}=\{y \in \mathbb R: y > z\}$ and positive in $D_{+}=\{y \in \mathbb R: b<y < z\}$ .
Substituting $\rho(y,t)$ into \eqref{eq7_adapt}, we have $\dot{V}<0$ in 
$D_S=\{y \in \mathbb R_{+}: y > b+e^{\frac{\eta}{\alpha}}\} \cup \{y \in \mathbb R_{-}: b < y < b-e^{\frac{\eta}{\alpha}}\}$ and 
$\dot{V} >0$ in 
$D_U=\{y \in \mathbb R_{+}: b < y < b-e^{\frac{\eta}{\alpha}}\} \cup \{y \in \mathbb R_{-}: y > b+e^{\frac{\eta}{\alpha}}\}$.
Also, $\rho(y,t) \to -\infty$ for $y \to b+0$.
This gives us the problem of sliding along the surface $b$.
Fig.~\ref{Fig_ad_4} (right) shows the transients for $\alpha=10$, $b=2e^{-0.1t}+0.5\sin(0.5t)-0.5$, and $d(t)=5+5\sin(7t)$.



\section{Conclusion}
\label{Sec6}

The paper considers density systems depending on the density function on the right hand side of the equation. 
The density function specifies the properties of the space and influences the behaviour of the considered systems.
This property is also used in the design of control laws.
Depending on the type of the density function, it is possible to obtain both classical and new control laws, which ensure a  new set of goals.
In particular, the adaptive control law is given with a guarantee of transients in certain domains, while the classical adaptive control scheme provides only the practical ultimate boundedness.
In this case, the parameters of the given domains are set using the density function by setting the density of the space.
The simulation results confirm the theoretical conclusions.

\section*{Acknowledgment}
The work was carried out at the IPME RAS 
with the support of goszadanie no. FFNF-2024-0008 
(no. 124041100006-1 in EGISU NIOKTR).

\end{document}